\documentclass[12pt]{article}

\usepackage{amsfonts,epsfig}

{\makeatletter
 \gdef\xxxmark{%
   \expandafter\ifx\csname @mpargs\endcsname\relax 
     \expandafter\ifx\csname @captype\endcsname\relax 
       \marginpar{xxx}
     \else
       xxx 
     \fi
   \else
     xxx 
   \fi}
 \gdef\xxx{\@ifnextchar[\xxx@lab\xxx@nolab}
 \long\gdef\xxx@lab[#1]#2{{\bf [\xxxmark #2 ---{\sc #1}]}}
 \long\gdef\xxx@nolab#1{{\bf [\xxxmark #1]}}
}

\newtheorem{theorem}{Theorem}[section]
\newtheorem{lemma}[theorem]{Lemma}
\newtheorem{claim}[theorem]{Claim}
\newtheorem{corollary}[theorem]{Corollary}

\newtheorem{definition}[theorem]{Definition}

\newtheorem{conjecture}[theorem]{Conjecture}
\newenvironment{proof}{\noindent{\bf Proof:\/}}{\hfill $\Box$ \vskip 0.1in}

\newcommand{\opt}{\textsc{opt}}
\newcommand{\eth}{\textsc{eth}}
\newcommand{\pcp}{\textsc{pcp}}
\newcommand{\pgc}{\textsc{pgc}}
\newcommand{\sat}{\textsc{sat}}
\newcommand{\setcover}{\textsc{setcover}}
\newcommand{\clique}{\textsc{clique}}
\newcommand{\mmis}{\textsc{mmis}}
\newcommand{\fpt}{\textsc{FPT}}
\newcommand{\subexp}{\textsc{subexp}}
\renewcommand{\P}{\textsc{p}}
\newcommand{\NP}{\textsc{np}}
\newcommand{\minrep}{\textsc{min-rep}}

\newcommand{\poly}{\mathrm{poly}}

\begin{document}

\title{Reducing the Optimum: Fixed Parameter Inapproximability for
$\clique$ and $\setcover$ in Time Super-exponential in $\opt$}

\author{%
{\sl Mohammad T. Hajiaghayi\/}\thanks{Department of Computer Science
, University of Maryland at College Park, USA. Supported in part by
NSF CAREER award 1053605, NSF grant CCF-1161626, ONR YIP award
N000141110662, and DARPA/AFOSR grant FA9550-12-1-0423. Email:
\texttt{hajiagha@cs.umd.edu}.}
  \medskip \and
   {\sl Rohit Khandekar\/}\thanks{
 KCG holdings Inc., USA,
  \texttt{rkhandekar@gmail.com}.}
\medskip \and
{\sl  Guy Kortsarz\/}%
\thanks{Department of Computer Science, Rutgers University-Camden,
USA. Supported in part by NSF grant 1218620.
\texttt{guyk@camden.rutgers.edu}.}}
\date{}
\maketitle
\vspace{-0.3in}

\begin{abstract}
A minimization (resp., maximization) problem is called fixed
parameter $\rho$-inapproximable, for a function $\rho \geq 1$, if
there does not exist an algorithm that given a problem instance $I$
with optimum value $\opt$ and an integer $k$, either finds a
feasible solution of value at most $\rho(k)\cdot k$ (resp., at least
$k/\rho(k)$) or finds a certificate that $k < \opt$ (resp., $k >
\opt$) in time $t(k)\cdot |I|^{O(1)}$ for some function $t$.  In
this paper, we present motivations for studying inapproximability in
terms of the parameter $\opt$, the optimum value of an instance.  A
problem is called $(r,t)$-$\fpt$-hard in parameter $\opt$ for
functions $r$ and $t$, if it admits no $r(\opt)$ approximation that
runs in time $t(\opt) |I|^{O(1)}$.  To prove hardness, we use gap
reductions from 3-$\sat$ and assume the Exponential Time Hypothesis
({\eth}). It is easy to see that if the value of $\opt$ is known in
the `yes' instance of the gap reduction, inapproximability w.r.t.
$\opt$ implies inapproximability w.r.t. input integer $k$. The
converse is not true. Hence inapproximability in $\opt$ is stronger.
Previous FPT-hardness results for the problems we study~\cite{CHK}
have running times $t$ that are sub-exponential in $\opt$. Such
results can often be obtained by a simple `translation' of
inapproximability results to FPT-hardness. In this paper, therefore,
we are only interested in times $t(\opt)$ that are super-exponential
in $\opt$. Fellows~\cite{fc} conjectured that $\setcover$ and
$\clique$ are $(r,t)$-$\fpt$-hard for {\em any} pair of
non-decreasing functions $r,t$ and input parameter $k$.  We give the
first inapproximability results for these problems with running
times {\em
  super-exponential in $\opt$}. Since one would like to prove
inapproximability with as fast growing functions $r,t$ as possible,
it is critical to reduce the value of $\opt$ (relative to the gap in
the $\opt$ value between yes and no instances and the size of the
instance). Our paper introduces systematic techniques to reduce the
value of the optimum.  These techniques are robust and work for
three quite different problems. In particular one of our results
shows that, under $\eth$, $\clique$ is $(r,t)$-$\fpt$-hard for
$r(\opt)=1/(1-\epsilon)$ with some constant $\epsilon>0$ and {\em
any} non-decreasing function $t$. The running time can be also set
to $2^{o(n)}$, for an arbitrary $o(n)$ exponent. This improves the
main result of Feige and Kilian~\cite{fk} in two ways. We also show
that the Minimum Maximal Independent Set ($\mmis$) problem is
$(r,t)$-$\fpt$-hard in $\opt$, for arbitrarily fast growing
functions $r,t$ of $\opt$.  While a similar result was known in
terms of parameter $k$, we present this result since the reduction
of $\opt$ value in this case, is very instructive and we believe
will find further applications.
\end{abstract}

\section{Introduction}


\subsection{FPT inapproximability with parameter $\opt$: motivation}

In {\em Fixed Parameter Tractability} ($\fpt$) theory, we are given
a decision problem $P$ with a parameter $k$, that relates to the
problem instance. An $\fpt$ algorithm for a problem is an exact
algorithm that runs in time $t(k)\cdot n^{O(1)}$ where $n$ is the
size of the instance and $t$ is an arbitrary
function.\footnote{Unless otherwise stated, all mentioned functions
  are total computable functions from non-negative integers to
  themselves.} An $\fpt$ {\em approximation} algorithm,
given a parameter $k$, approximates
the desired solution value within a ratio $r(k)$ and runs in time
$t(k)\cdot n^{O(1)}$, for a function $r$. More precisely, an
$\fpt$-approximation algorithm is defined as follows. For a
minimization (resp., maximization) problem P, an algorithm is called
an $(r,t)$-$\fpt$-approximation algorithm for $P$ with input parameter
$k$, if the algorithm takes as input an instance $I$ with (possibly
unknown) optimum value $\opt$ and an integer parameter $k$ and either
computes a feasible solution to $I$ with value at most $k\cdot r(k)$
(resp., at least $k/r(k)$) or computes a certificate that $k < \opt$
(resp., $k > \opt$) in time $t(k)\cdot |I|^{O(1)}$. In the latter
case, such a certificate can be obtained from the analysis of the
algorithm and the fact that it did not produce the desired
solution. Here the goal is to design algorithms with as {\em slow}
growing functions $r$ and $t$ as possible. A problem is called
$(r,t)$-$\fpt$-inapproximable (or, $(r,t)$-$\fpt$-hard) if it does not
admit any $(r,t)$-$\fpt$-approximation algorithm. Here the goal is to
show inapproximability with as {\em fast} growing functions $r$ and
$t$ as possible.

In this paper, we use reductions from $3$-$\sat$ to prove $\fpt$-hardness.
When doing a reduction from $3$-$\sat$ to an optimization problem $P$,
a ``yes'' instance of $P$ is an instance obtained from a satisfiable
formula, and a ``no'' instance of $P$ is an instance obtained from a
reduction of non-satisfiable formula.  Whenever we use reductions from
3-$\sat$, we denote the number of variables by $q$, the number of
clauses by $m$ and $N=m+q$.

Our work is motivated by a conjecture, by Mike Fellows, concerning
parameterized approximation for $\setcover$ and $\clique$.

\begin{conjecture}[FPT-hardness of $\setcover$ and $\clique$ (Fellows~\cite{fc})]
\label{conj1}
\hspace{1in}

\begin{itemize}
\item $\setcover$ is $(r,t)$-$\fpt$-hard for any non-decreasing
  functions $r$ and $t$.
\item $\clique$ is $(r,t)$-$\fpt$-hard for any non-decreasing
  functions $r(k)=o(k)$ and $t$.
\end{itemize}
\end{conjecture}

When we do a gap reduction from $3$-$\sat$, in all cases in this
paper, the optimum $\opt$ of a yes instance is known. Thus, by
definition, inapproximability in terms of $\opt$ implies
inapproximability in terms of $k$ by setting $k=opt$. On the other
hand, inaproximability in terms of $k$ does not imply
inapproximability in $\opt$ because it may be that $k<\opt$ but we
can not give a certificate that $k<\opt$. Typically, this happens if
$k$ is only slightly smaller than $\opt$. In this case, proving
 $\rho(k)$-$\fpt$ hardness, implies that
is it is not possible to find a solution of size $k\cdot \rho(k)$.
Clearly, showing that we cannot find a solution of value
$\opt\cdot \rho(\opt)$ is a stronger
statement if $k<opt$. Thus if $\opt$ is known
inapproximability in $\opt$ implies inapproximability in $k$,
but not vise-versa.

Therefore, we suggest a principle of
studding the hardness in terms of $\opt$
whenever possible. The following version of the conjecture, is with
parameter $\opt$.

\begin{conjecture}[FPT-hardness of $\setcover$ and $\clique$ with parameter
OPT.]
\label{conj2}
Let $\opt$ and $n$ denote the value of the optimum and size of the
given instance, respectively.
\begin{itemize}
\item $\setcover$ admits no $r(\opt)$ approximation that runs in time
$t(\opt)\cdot n^{O(1)}$ for any non-decreasing functions $r$ and
$t$.
\item $\clique$ admits no $r(\opt)$ approximation that runs in time
$t(\opt)\cdot n^{O(1)}$ for any non-decreasing functions $r(k)=o(k)$
and $t$.
\end{itemize}
\end{conjecture}

\subsection{Reduction with sub-exponential time in $\opt$ and their weaknesses}

Many times it is automatic to translate an inapproximability result,
to $\fpt$-hardness. This happens if we are allowed to prove
$(r,t)$-$\fpt$-hardness for time $t(opt)$, subexponential in $\opt$.

Consider a polynomial time gap reduction from
$3$-$\sat$ to any other problem $P$.
Let $|I|=q+m$ be the size of the $3$-$\sat$ instance
and thus the size of the instance
of $P$ is clearly bounded by $m^c$, for some constant $c$.

Many times the gap with respect to $\opt$ and the gap with respect to $n$,
are about the same, as
$\opt$ is close to $n$. For
example, for all standard reduction from $3$$\sat$
to $\clique$ and $\setcover$, $n$ and $\opt$ are very close.
However, there are even cases in which $\opt$ is much smaller than $n$.
Thus a gap in $n$ implies a larger gap in $\opt$.

In any case, say that we have $\rho(\opt)$
gap that resulted from a reduction from $3$-$\sat$.
Thus the $\eth$ implies that we cannot
find an approximation better than
$\rho(\opt)$ in time $2^{o(m)}$.

Thus, all we need to do is to translate $2^{o(m)}$ to a function of
$\opt$. In almost all cases $\opt\leq n$ (albeit if $\opt$ is
polynomial in $n$ the next claim still holds). In that case,
$\opt\leq n=m^c$. For any constant $c'>c$, $2^{\opt^{1/c'}}=
2^{o(m)}$. Thus we automatically get a $\rho(\opt)$-$\fpt$-hardness
in time $2^{\opt^{1/c'}}$. Clearly, the meaning of such a reduction
is limited and it is a translation of the hardness result to
$\fpt$-hardness language. The detail that allows us to give such a
translation is that  $t(\opt)$ is subexponential in $\opt$.

\subsection{Reducing the value of OPT}

We are not able to prove either of these conjectures with our current
techniques. It seems to us that the current state of knowledge is not enough to
prove the conjectures. In fact we suspect that in order to prove this
conjecture, a parameterized version of the $\pcp$ is needed.
Nevertheless, we make a very important breakthrough proving hardness
results with running times $t(\opt)$ that are {\em super-exponential}
in $\opt$. Such results were never proven before. In particular, the
inapproximability of \cite{CHK} is under a {\em sub-exponential
  running} time.  We discuss this and the relation to the reducing the
value of $\opt$.
Proving hardness result with super-exponential time in $\opt$ or $k$
necessitates reducing the value of $\opt$. The reason is that
$t(\opt)=2^{o(m)}$ must hold with $m$ the number of clauses in the
$3$-$\sat$ instance we reduce from. As the function $t$ becomes faster
and faster growing, $\opt$ needs to get smaller and smaller.
Therefore we claim that an important aspect of
the art of proving $\fpt$-hardness is
creating new instances with smaller and smaller $\opt$.

We develop systematic techniques to reduce the value of $\opt$ for the problems
we study, with the following property.  Given a pair of `yes' and `no'
instances for the problem, with optimum values $\opt_y$ and $\opt_n$
respectively, the reduction creates new instances with optimum values
$\opt'_y$ and $\opt'_n$ respectively such that
\begin{itemize}
\item the new optimum values are much smaller: $\opt'_y \ll \opt_y$ and
  $\opt'_n \ll \opt_n$,
\item the gap between optimum values in yes and no instances is
  preserved: $\opt'_y/\opt'_n \approx \opt_y/\opt_n$,
\item the new instance sizes are not much larger than the old ones.
\end{itemize}
No {\fpt} hardness before us used any techniques to reduce $\opt$,
even though reducing the $\opt$ seems to us to be one of the most
natural and important ideas for proving ${\fpt}$-hardness.

Our approach can be summarized as follows.
\begin{enumerate}
\item  Try to prove inapproximability in terms of $\opt$ and not in terms of $k$.
For this you need gap reductions in which the value of a ``yes''
instance is known (it is also possible to do reductions in terms of $\opt$
if its only approximately known but its much better if $\opt$ is known).
\item Prove only inapproximability in time super-exponential in $\opt$
or larger.
\item An important tool in proving a  strong ${\fpt}$-hardness
is developing techniques to reduce an instance of a problem
to another instance of the same problem, with much smaller
optimum, while not loosing much in the gap.
\end{enumerate}

Apart from $\setcover$ and $\clique$, we also study a problem called
the {\em Minimum Maximal Independent Set ($\mmis$)} problem. In this
problem, given an undirected graph, the goal is to find a minimum-size
independent set that is also inclusion-wise maximal.  This is also
equivalent to finding the minimum-size independent set that is also a
dominating set.


\section{Previous work}

The following relation is known among the parameterized complexity
classes: $\fpt \subseteq W[1]\subseteq W[2]$. It is widely believed
that $\fpt\neq$ W[1]. In fact $\fpt$=W[1] implies that $\eth$ fails.

{\bf Inapproximability that is sub-exponential in $n$:} A widely
explored line of research shows, for $\clique$ and $\setcover$, a
relation between an approximation, and the running time it
requires. Such results are discussed in \cite{subexp}.  Recently,
\cite{focs13} improved \cite{subexp} to get the following result.  For
any $r$ larger than some constant and any constant $\epsilon>0$, any
$r$-approximation algorithm for the maximum independent set problem
must run in time at least $2^{n^{1-\epsilon}/r^{1+\epsilon}}$. This
nearly matches the upper bound of $2^{n/r}$.  In this case super
exponential running times are out of the question, because the time
depends on $n$.  This again shows the power of parameterizing
algorithms. In the instance we start with, the optimum is very close
to $n$.  By reducing $\opt$ we can get inapproximability for $\clique$
and $\setcover$ in time super-exponential in $\opt$, giving a more
refined classification of the problems.

To the best of our knowledge, the effort of showing
$\fpt$-hardness for $\clique$ and $\setcover$
(in terms of $k$ and $\opt$) started with \cite{CHK}.
\begin{theorem}{\em \cite{CHK}}
Under  $\eth$ and $\pgc$, there exist constants $1> F1, F2 > 0$ such
that the $\setcover$ problem does not admit an $\fpt$ approximation
algorithm with ratio $\opt^{F_1}$ in time $2^{\opt^{F_2}}$.
\end{theorem}

The above theorem uses $F_2<1$ hence uses time sub-exponential
time $\opt$ and is not suited for this paper.

\begin{theorem}{\em \cite{CHK}}
Unless $\NP \subseteq \subexp$, for every $0<\delta<1$ there exists
a constant $F=F(\delta) > 0$ such that $\clique$ admits no $\fpt$
approximation within $\opt^{1-\delta}$ in time $2^{\opt^F}$.
\end{theorem}

As $F<1$ in the above construction the running time here too is
sub-exponential in $\opt$ and this theorem is not suited for this
paper.

These theorems seem unrelated to the results of \cite{subexp} and
\cite{focs13} because of the large inapproximability proved in the
theorems, compared to the constant lower bound in the above papers.

{\bf Previous work on} $\mmis$: The problem is known to be W[2]-Hard
\cite{mmis}.  It is also known that it is $(r(k),t(k))$-$\fpt$-hard
for any increasing $r,t$.  The $r(k)$ inapproximability for any $r$
was shown in \cite{dis}, under the assumption that $\fpt\neq W[2]$.
Indeed, reducing directly from $W[1]$ or $W[2]$ creating a gap,
is another important technique to get $\fpt$-hardness
(if somewhat hard to use).
The arbitrary time $t(k)$ follows from proving that the {\eth} implies
$\fpt\neq W[2]$.  This proof is covered by the papers
\cite{Chen,fel,Dow}
\footnote{
The version we discuss is from the survey paper
http://www.cs.bme.hu/~dmarx/papers/survey-eth-beatcs.pdf.}.

The proof in that survey first shows that if Clique admits an {\fpt}
algorithm, then $3$-coloring admits a $2^{o(m)}$ time algorith
which implies that $3$-$\sat$ admits a $2^{o(m)}$ exact solution.
This is due
to the simple linear reduction from $3$-$\sat$ to $3$-Coloring. This in
turn implies that the {\eth} fails. Thus assuming {\eth}, Clique
admits no {\fpt} algorithm.  Then it is shown that if Clique has no
$\fpt$ algorithm a problem called {\em Multicolored Clique} admits no
$\fpt$ algorithm.  Finally it is proved that the if the Multicolored
Clique problem admits no $\fpt$ algorithm, the Dominating Set (hence
$\setcover$ as well) admits no $\fpt$ algorithm.

We prove a very slightly stronger statement namely
$(r(\opt),t(\opt))$-$\fpt$-hard for any increasing $r,s$
with the optimum known. However, this
slight improvement is not the reason we include that proof. The
reduction of the optimum in this case is quite non-trivial and
requires ideas that probably will find future applications.
To understand the difficulty
consider a reduction from $3$-SAT
to
$\mmis$ in which we basically try to transform
the instance into one in which literals (that will
be the minimal independent set)
will dominate
clauses that they belong to.  A great difficulty is
that we need to choose a small number of the literals, only.
The literals not chosen must be dominated as well.
Thus, we need to add edges between literals. But
then: how do we make sure that the independent set
has no edges between literals in the set, and on top of that make the optimum much smaller?


In \cite{CHK} a large collection of $W[1]$-hard problems are
presented for which an inapproximability such as in the Fellows
conjecture does not apply. In fact, these problems are given some
approximation $f(\opt)$ for $f(x)\leq x^2$ and the running time is
just polynomial in the size of the input. All these problems are not
only $W[1]$-hard, but also admit strong inapproximability results
(at least Label-Cover hardness or believed to have no better than
polynomial ratio like the minimization version of the Dense
$k$-subgraph problem).

While all problems of \cite{CHK} are minimization problems, such
results hold for maximization problems as well, as the next example
shows.  For a set $U$, the edges $e(U)$ are the edges with both
endpoints in $U$. The parameterized version of the {\em Dense
  $q$-subgraph} is defined as follows. The input is a simple connected
graph $G(V,E)$ and parameters $q$ and $p$.  The question is whether
there is a set $U$ with $q$ vertices so that, $e(U)$ is at least $p$?
Clearly this problem is $W[1]$-hard as $\clique$ is a special case of
it. But returning a spanning tree on any $q$ vertices gives $\opt+3$
approximation. To see that note that the number of edges in such a
tree is $q-1$. In addition, $p<q^2$ as if $p>q^2$ we can say no
immediately as no set with $q$ vertices and $q^2$ edges exits. In
other words, $q-1\leq \opt\leq q^2$ for the instance. Since the number
of edges returned is $q-1$, the ratio is $p/(q-1)\leq q^2/(q-1)\leq
q+2$. This implies a ratio of $\opt+3$ as $\opt\geq q-1$. The time is
polynomial in $n=|V|$.  For future reference, the minimization version
of the Dense $k$-subgraph problem admits as an input a connected graph
and two parameters $Q$ and $p$. The question is if there exists a set
$U$ of size $p$, so that $e(U)>Q$.

\begin{theorem}{\em \cite{CHK}}
Dense $k$-subgraph, Directed Multicut, Directed Steiner Tree,
Directed Steiner Forest, Directed Steiner Network and  the minimization version of the
Dense $k$-subgraph problem, admit  $g(\opt)$-approximation algorithms
that runs in polynomial time, for some small
function $g$ (the largest approximation ratio we give is $\opt^2$).
\end{theorem}

\noindent
{\bf Approximation under $\fpt$ time:}
The Strongly connected Steiner Subgraph problem
is given a graph $G(V,U)$ with a set of terminals
$T\subseteq V$, find a strongly connected subgraph that
includes all the terminals. It is elementary to see
that this problem is equivalent with respect to approximation
to the Directed Steiner Tree problem that admits no better than
$\log^{2-\epsilon} n$ ratio for any constant $\epsilon$ \cite{eran1}.
\begin{theorem}{\em \cite{CHK}}
Strongly Connected Steiner Subgraph problem is W[1]-hard and does
not admits a better approximation than $\log^{2-\epsilon} n$ for any
constant $\epsilon$. However, if $\fpt$ time is allowed we can get a
2-approximation algorithm with running time $h(\opt)n^{O(1)}$.
\end{theorem}

So far this seems to be the only natural problem for which such a
result is known where allowing $FPT$ times reduces the order of
magnitude the approximation factor. This gives evidence that we
should try to find $\fpt$ approximations for $W[1]$ or $W[2]$ hard
problems that also admit a strong inapproximability. Maybe in FPT time,
a better approximation is possible?

\noindent
{\bf Previous work from inapproximability theory:}

\noindent{\bf Notation:} Throughout this paper, the number of sets
in a $\setcover$ instance and the number of vertices in a $\clique$
instance will be denoted by $n$.

\begin{theorem}{\em \cite{hast,ZC}}
\label{thm:clique} Unless $\P =\NP$, $\clique$ can not be
approximated within $n^{1-\epsilon}$.
\end{theorem}
The reduction in \cite{hast} is randomized but in \cite{ZC} this
result is derandomized and it is achieved under the assumption that
$\P \neq \NP$. Note that there is a stronger inapproximability
result for $\clique$, e.g., in \cite{SK}, but we can not use it
because the running time of the reduction is quasi polynomial (in
fact if the reduction would have been polynomial, this would
contradict $\eth$).

\begin{theorem}
\label{thm:setcover} {\em \cite{RS}} Unless $\P =\NP$, $\setcover$
admits no better than $c\log n$ approximation for some constant $c$.
\end{theorem}

We note that the inapproximability results of
Theorems~\ref{thm:clique} and \ref{thm:setcover} are close to the
best possible algorithmic results. $\clique$ admits  a trivial $n$
ratio approximation. For $\setcover$, it is known (folklore) that
the natural greedy algorithm gives a $\ln n+1$ approximation.

\section{Preliminaries}
\label{sec:prelim}

We begin by describing the complexity theoretic conjectures assumed
in proving the hardness results in this paper.  Impagliazzo et
al.~\cite{eth-paturi} formulated the following conjecture which is
known as $\eth$. We assume $\eth$ in all hardness results in this
paper.
\begin{center}
\noindent
\framebox{\begin{minipage}{\textwidth} \underline{Exponential Time
Hypothesis ($\eth$)}\\
3-$\sat$ cannot be solved in
$2^{o(q)}(q+m)^{O(1)}$ time where $q$ is the number of variables and
$m$ is the number of clauses.
\end{minipage}}
\end{center}

Using the Sparsification Lemma of Calabro et al.
\cite{sparsification-paturi}, the following lemma follows.
\begin{lemma}
\label{lem:sparsification-eth}
Assuming $\eth$, 3-$\sat$ cannot be solved in $2^{o(m)}(q+m)^{O(1)}$
time where $q$ is the number of variables and $m$ is the number of
clauses.
\end{lemma}
{\bf Remark on the relation between $N$ and $m$:}
We may assume without loss of generality that
there are more clauses than variables in the
$\sat$ instance. The other case is similar.
Thus if the number of variables is $q$ we get
$N=q+m\leq 2m$.
Hence, we do not need to use $N$, and $m$ can replace it
in any future computation (the factor $2$ difference is never significant in this paper).

Another conjecture used is Projection Game Conjecture ($\pgc$) due
to Moshkovitz~\cite{r3}.

\begin{center}
\noindent
\framebox{\begin{minipage}{\textwidth} \underline{Projection Game
Conjecture ($\pgc$)}\\
There exists some constant $c$ and $\pcp$ of size $m\cdot
\rho(m)\cdot \poly(1/\epsilon)$, with soundness $1/\epsilon$ for any
$\epsilon$ so that $1/\epsilon\leq m^c$. The alphabet is of size
$\poly(1/\epsilon)$ and $\rho(m)=2^{\log^{\alpha} m}$, for a
constant $\alpha$, $0<\alpha<1$. The $\pcp$ is equivalent to the
Labelcover problem with the projection property.
\end{minipage}}
\end{center}

In \cite{r3} the sublinear term $\rho(m)$ was not discussed and thus
we took the sublinear term from \cite{dana}. The conjecture is
actually already {\em proven} in \cite{dana} if we allow alphabet
exponential in $1/\epsilon$.

The $\minrep$ problem was first defined in \cite{guylower}
where it is shown that if Labelcover (see \cite{guylower})
has gap $1/\epsilon$, $\minrep$ has gap $\sqrt{1/\epsilon}$.
The $\minrep$ problem is a problem in graphs and is easy to understand,
which is the reason it is defined as such in \cite{guylower}.

\begin{center}
\framebox{\begin{minipage}{\textwidth} \underline{Min-Rep
($\minrep$)}\\
The instance is a bipartite graph $G=(U,W,E)$ where $U$ and $W$ are
each split into a disjoint union of $q$ sets $U=\cup_{i=1}^q A_i$
and $W=\cup_{i=1}^q B_i$. The sets $A_i$ and $B_i$ are called
supervertices. The graph $G$ and the partition into supervertices
induce edges $E^+$ on supervertices: for two supervertices $A_i$
and $B_j$, we have superedge $(A_i,B_j) \in E^+$ iff there exist
$a\in A_i$ and $b\in B_j$ such that $(a,b)\in E$. We say that a
subset $S \subseteq U\cup W$ {\em covers} a superedge $(A_i,B_j)
\in E^+$ if there exist $a\in A_i\cap S$ and $b\in B_j\cap S$ so
that $(a,b)\in E$. The goal in $\minrep$ is to find a minimum-size
set $S \subseteq U\cup W$ that covers all superedges in $E^+$.
\end{minipage}}
\end{center}

The result of \cite{guylower,r3}, imply
the following reduction from {\pgc} to $\minrep$.
The $\pcp$ in question can be posed as
a Min-Rep instance because the $\pcp$ is queried
in two places.
The number of supervertices in the supergraph
equals the size of the $\pcp$.
Thus by \cite{r3}, the number of questions
in the $\minrep$ instance is
$O(m\cdot \exp(\log^\alpha m) \cdot \poly\log(m) \cdot
(1/\epsilon)^{c_1})$ for some constant $c_1$.
Because we deal with the {\pgc}, we may assume that
the number of vertices that belong to any supervertex of the $\minrep$ graph
is $(1/\epsilon)^{c_2}$ for some constant $c_2$.
Thus the size of the $\minrep$ graph, namely,
the number of vertices is
$O(m\cdot \exp(\log^\alpha m) \cdot \poly\log(m) \cdot
(1/\epsilon)^{c_1+c_2})$.

In a ``yes'' instance, it is possible to choose one vertex
per supervertex and cover all the superedges, thus the
optimum is
$O(m\cdot \exp(\log^\alpha m) \cdot \poly\log(m) \cdot
1/\epsilon^{c_1})$.

\begin{center}
\framebox{\begin{minipage}{\textwidth} \underline{Reduction from $3$-$\sat$ to
$\minrep$~\cite{guylower,r3} }\\  
Given an instance $I$ of 3-$\sat$ with size $m$ equal to the number
of clauses in $I$, for constants $c_1,c_2,c_3,\alpha>0$ and any
choice of $\epsilon$ so that $\epsilon > 1/m^{c_3}$, there exists a
`yes'' instance of $\minrep$ which admits a feasible solution of
size $\sigma = O(m\cdot \exp(\log^\alpha m) \cdot \poly\log(m) \cdot
1/\epsilon^{c_1})$. The optimum for a ``no'' instance of $\minrep$
is at least $\Omega(1/\sqrt{\epsilon})$ times the value of a ``yes''
instance.
The projection property is translated here to the
following. The graph induced by any $A_i\cup B_j$ that are a super
edge is a union of disjoint stars with heads in $A_i$.
\end{minipage}}
\end{center}

By composing the above reduction with a reduction to $\minrep$
of \cite{guylower}, with the reduction
from $\minrep$ to $\setcover$ \cite{LLYY},
the following corollary is derived:

\begin{corollary}
\label{scc}
There exists a reduction
from $3$-$\sat$ to $\setcover$ with $m$ clauses so that
\begin{enumerate}
\item The number of sets is $m\cdot 2^{\log^{\alpha} m}\cdot \poly\log(m)\cdot
(1/\epsilon)^{c_1+c_2}$ with $c_1$ and $c_2$ some constants larger than $0$
and $\alpha$ a constant that satisfies $0<\alpha<1$.
\item The number of elements is $\poly(m)$.
\item
The optimum for a ``yes'' instance is exactly $m\cdot
2^{\log^{\alpha} m}\cdot \poly\log(m)(1/\epsilon)^{c_1}$.
\item
The optimum for a ``no'' instance is at least $d\cdot
\sqrt{\epsilon}\cdot  m\cdot 2^{\log^{\alpha} m} \cdot
\poly\log(m)\cdot (1/\epsilon)^{c_1}$ for a constant $d>0$.
\end{enumerate}
\end{corollary}

As we shall later see, we
choose $\epsilon=c^2/\log^2 m$ for a constant $c>0$. The
inapproximability of $\setcover$ becomes $c\cdot \log m$. In
addition, since $1/\epsilon$ is by itself polylogarithmic in $m$, we
can denote the size of the optimum for a ``yes'' instance by $m\cdot
2^{\log^\alpha m}\cdot \poly\log(m)$.

We briefly explain how does the \cite{LLYY} reduction works.

Every vertex in the ${\minrep}$ instance is a set in the reduction
of \cite{LLYY}.
As the size of the $\minrep$ graph is
$m\cdot 2^{\log^{\alpha} m}\cdot
\poly\log(m) (1/\epsilon)^{c_1+c_2} =m\cdot 2^{\log^{\alpha} m}\cdot
\poly\log(m)$, this is also the number of sets in the $\setcover$ instance.
For every superedge $A_i,B_j$ we add a set $M_{i,j}$
of elements.
Note that the number of superedges is no larger than
the number of supervertices pairs, and so is no larger than $m^3$.
The size of every $M_{i,j}$ is
$(1/\epsilon)^{c_2}=\poly(m)$ and
the total number of elements
is $m^3\cdot \poly(m)=\poly(m)$.

In \cite{LLYY} every $a_i\in A_i,b_j\in B_j$ so that $(a_i,b_j)\in
E$ are connected to two disjoint random halves of $M_{ij}$ (the
reduction we described is randomized for simplicity of the
explanation. However, this construction can be derandomized). For a
``yes'' instance we can choose a single vertex out of every
supervertex and cover all superedges (see both \cite{LLYY} and
\cite{r3}). Thus for a ``yes'' instance we can pick a $\setcover$ of
size $m\cdot 2^{\log^{\alpha} m}\cdot \poly\log(m)$, because all
superedges are covered, and thus there will be $a_i,b_j$ so that
$a_i$ covers some random half and $b_j$ covers the other half.

The $\sqrt{1/\epsilon}$ gap for $\minrep$ implies that in a ``no''
instance, unless you take $\Omega(\log m)$ times the optimum value
for a ``yes'' instance sets, you hardly cover any of the superedges.
We choose $|M_{ij}|=2^{c\log m}$ for some constant $c$. In the ``no'' instance, since most
superedges will not be covered (there will not be $a_i,b_j$ in the
solution so that $b_j$ belongs to the star of $a_i$) $M_{ij}$ has to
be covered by a collection of random independent sets of size
$|M_{ij}|/2=2^{c\log m}/2$. Using random halves, about $c\cdot \log
m$ sets are required to cover a single $M_{ij}$ which is the source
of the gap.

\section{Our results}
\label{sec:our-results}

Recall that we call an optimization problem $(r,t)$-$\fpt$-hard if it
admits no algorithm with approximation ratio $r(\opt)$ and running
time $t(\opt) \cdot n^{O(1)}$, where $\opt$ is the optimum
of some instance, $n$ is the
size of the given instance and $r,t$ are given functions.
In all our reductions $\opt$ is known and so the reduction
implies a similar hardness for a parameter $k$.
\begin{theorem}
\label{1} Under  $\eth$ and $\pgc$, $\setcover$ is
$(r,t)$-$\fpt$-hard for $r(\opt) = (\log \opt)^\gamma$ and $t(\opt)
= \exp(\exp((\log \opt)^\gamma))=\exp\left (\opt^{(\log^f
\opt)}\right)$ for some constant $\gamma>1$ and $f=\gamma-1$.
\end{theorem}
The time here is much larger than just exponential in $\opt$.

\begin{theorem}
  \label{3}
  Under $\eth$ and a stronger version of $\pgc$ with $\pcp$ length
    $O(m\cdot \poly\log(m) \cdot \log(1/\epsilon))$ and gap
    $\Omega(1/\epsilon)$ for any $\epsilon\geq 1/m^c$, for some
    constant $c$, $\setcover$ is $(r,t)$-$\fpt$-hard for $r(\opt) =
    \opt^{d'}$ and $t(\opt) = \exp(\exp(\opt^{d''}))$ for some
    constants $d',d''>0$.
\end{theorem}
This kind of $\pcp$ was conjecture to exist by Moshkovitz in a private
communication.

Note that the running times in this result is almost doubly
exponential in $\opt$. We later show  this result is basically the
best we can expect if we just use $\pcp$ (even under the best
conceivable $\pcp$).

We can also prove an inapproximability with super-exponential time
in $\opt$ that only assumes $\eth$.
\begin{theorem}
\label{alone} Under $\eth$ alone, $\setcover$ cannot be approximated
within $$c\sqrt{\log \opt}$$ for some constant $c$, in time
$\exp\left (\opt^{(\log \opt)^{f}}\right )$ for $f$
the same constant from Theorem \ref{1}.
\end{theorem}
Here the inapproximability we get is significantly smaller if we can
not assume the $\pgc$. But the running time is the same,
hence super-exponential.

\begin{theorem}
\label{2} Under $\eth$, $\clique$ is $(r,t)$-$\fpt$-hard for
$r(\opt) = 1/(1-\epsilon)$ for some constant $\epsilon$, that
satisfies $0<\epsilon<1$, and {\em any} non-decreasing function $t$,
however huge. The running time can also be set to $2^{o(n)}$ of our
choice of $o(n)$.
\end{theorem}
It is interesting to compare this result to the paper by Feige et al
 \cite{fk}.
In \cite{fk} it is shown that for $\opt\leq \log n$, $\clique$ can
not be solved exactly in time significantly smaller than
$n^{\opt}<n^{\log n}$ time, unless $\NP$ has sub-exponential
simulation. In fact the statement is much stronger than that. If
$\clique$ can be solved in much smaller time than $n^{\opt}$, any
solution for an $\NP$ problem that makes $f(n)$ non-deterministic
time, implies a deterministic solution in time {\em roughly}
$\exp(\sqrt{f(n)})$. This implies a host of NPC problems admits a
sub-exponential algorithm including $3$-$\sat$ (the number of non deterministic
bits used in $3$-$\sat$ is clearly at most $n$). Therefore it is
possible to prove \cite{fk} under {\eth} (which is our standard
assumption in this paper) as the contradiction in \cite{fk} implies
that {\eth} is not valid.

Theorem \ref{2} works for any $\opt$ and $\opt\leq \log n$ in
particular, and thus improves the paper of Feige et al \cite{fk} in
two ways.
First we prove $1/(1-\epsilon)$-hardness which for such
small values of $\opt$ might be significantly harder than ruling out
an exact solution.
Second, the $r(\opt)$-hardness holds
even if we allow time $2^{o(n)}$ time. This is a time much much larger
than $n^{\log n}$ of \cite{fk}.

{\bf Clarification:} It may seem strange that given time $2^{o(n)}$,
it is still not enough to exhaustively search for an optimum clique.
However, as a first step of our algorithm a graph of size $2^{o(n)}$
is constructed. Searching for the optimum exhaustively in such a
graph requires exponential in $n$ time and does not contradict
{\eth}.

\begin{theorem}
\label{22} Let
$$r(\opt)=\left(\frac{1}{1-\epsilon}\right)^{\log^{1/3} \opt},$$
with $\epsilon$ the constant from Theorem \ref{2}. Then $\clique$ is
$(r,t)$-$\fpt$-hard. for any function $t$, however huge.
\end{theorem}

As a function of $n$, we later show that
the time can be set to $2^{n^{1/Q(n)}}$ for an
arbitrarily slowly growing $Q(n)$. Thus Theorem \ref{22} improves
\cite{fk} in the same two ways that we mentioned for Theorem \ref{2}
. The inapproximability is now super constant, versus an exact solution,
and the running time is still
much much higher than $n^{\log n}$.

We study a well known $W[2]$ hard problem $\mmis$, and give
hardness in terms of $\opt$.

\begin{theorem}
\label{4} Under $\eth$, $\mmis$ is $(r,t)$-$\fpt$-hard in $\opt$
(and thus in $k$ since $\opt$ is known) for  any non-decreasing
functions $r$ and $t$.
\end{theorem}

It is
our opinion that one must try to prove as many such results to as
possible. This may shade light
on how to prove this for $\clique$ and $\setcover$.

\noindent {\bf Avoiding a constant optimum:}
Another standard we think is good imposing
in case of $\fpt$-hardness is that
the optimum for a "yes'' instance is not constant.
To explain that, we consider the
$3$-Coloring problem and (trivially) show that Fellows conjecture holds
for this problem. In
 \cite{condi} it is shown that 3-coloring admits no
constant approximation for {\em any} constant unless a {\em variant}
of the (well known) Khot two-to-one $\pgc$ holds. Take the ``yes''
instance of the problem (in which $\opt=3$). For any $r$,
$r(\opt)=g(3)$ is a constant and for any $t$, $t(\opt)\cdot
n^{O(1)}$ is just polynomial time. By the above hardness of
\cite{condi}, it is clear that for any $r,t$ $3$-coloring
is $(r,t)$-$\fpt$-hard.
The lesson to be derived from this example is that we
should only try hardness for problems with
non-constant $\opt$.

\subsection{Reducing the value of OPT while proving $\fpt$-hardness}

Our proofs are based on gap-reductions from $3$-$\sat$ to the given
optimization problem. To describe our technique, we define a notion of
gap reductions as follows. Fix a minimization problem $P$ and two
non-decreasing functions $r$ and $t$.
We use $m$ for the number of clauses in
the $3$-$\sat$ instance we reduce from.
\begin{definition}[Gap reduction]
\label{def:gapred} Let $m$ be the number of clauses in the
$3$-$\sat$ problem we reduce from. An algorithm that maps any given
instance $I$ of $3$-$\sat$ to an instance $M_I$ of problem $P$ is
called a gap-reduction with non-decreasing functions $r$ and $t$ if
there exists computable functions $\kappa_T$ and $\kappa_F$ such
that:
\begin{itemize}
\item $t(\kappa_T(m)) = 2^{o(m)}$ and $\kappa_T(m) \cdot
r(\kappa_T(m)) < \kappa_F(m)$,
\item The algorithm takes $2^{o(m)}$ time to construct $M_I$ (and
  hence the size of $M_I$ is $2^{o(m)}$),
\item $\opt(M_I) \leq \kappa_T(m)$ if $I$ is satisfiable,
\item $\opt(M_I) \geq \kappa_F(m)$ if $I$ is unsatisfiable.
\end{itemize}
\end{definition}

We now prove the following simple but useful claim.
\begin{claim}
\label{trivial1} If there exists a gap-reduction (according to
Definition~\ref{def:gapred}) for a minimization problem $P$ with
non-decreasing functions $r$ and $t$, then assuming $\eth$, problem
$P$ is $(r,t)$-$\fpt$-hard with parameter $\opt$.
\end{claim}
\begin{proof}
Assume on the contrary that there is an $r(\opt)$-approximation
algorithm $A$ with running time $t(\opt) \cdot m^c$ for $P$ where
$\opt$ is the value of the optimum, $m$ is the size of the given
instance of $P$ and $c > 0$ is a constant. Now we design an
algorithm for $3$-$\sat$ in time $2^{o(m)}$ where $m$ is the number
of clauses, as follows. Given a $3$-$\sat$ instance $I$, we first
use the gap-reduction to compute instance $M_I$ of $P$ and then run
algorithm $A$ on $M_I$ for time $t(\kappa_T(m)) \cdot m^c$ where $m$
is the size of $M_I$. Since $A$ is an $r(\opt)$-approximation, we
can decide whether $\opt(M_I) \leq \kappa_T(m) \cdot r(\kappa_T(m))
< \kappa_F(m)$. Thus we can decide whether $I$ is satisfiable. Note
that $A$ takes time $2^{o(m)}$, contradicting $\eth$. This finishes
the proof.
\end{proof}

\section{Inapproximability for Set Cover with super-exponential time in OPT}
\label{sec:lower-bound-set-cover} In this section we
prove Theorem~\ref{1}.

Corollary \ref{scc} implies the following corollary.
\begin{corollary}
\label{setcov} There exists constants $c_1,c_2>0$ and a constant
$0<\alpha<1$, so that the following holds. Let $m$ be the number of
clauses in the $3$-$\sat$ problem we reduce from. Assuming $\pgc$ and
$\eth$, there exists a reduction from $3$-$\sat$ to $\setcover$ so
that the number of sets in the resulting instance is $\sigma = m\cdot
2^{\log^{\alpha} m}\cdot \poly\log(m)\cdot (1/\epsilon)^{c_1+c_2}$ for
two constant $c_1,c_2>1$.  Furthermore, value of the optimum in
``yes'' instance is {\em exactly} $\kappa = m \cdot 2^{\log^{\alpha}
m}\cdot \poly\log(m)(1/\epsilon)^{c_1}$ and that in the ``no'' instance
is at least $c\cdot \log m\cdot \kappa$.
\end{corollary}
We use $\epsilon = c^2/\log^2 m$ here.

We now describe a way to change the $\setcover$ instance so that we
can use Claim~\ref{trivial1}.  The idea is to make the optimum much
smaller. Starting with the $\setcover$ instance $S=(U,\cal S)$ in the
above corollary, where $U$ is the set of elements and ${\cal S}
\subseteq 2^U$ is the collection of sets, we construct a new instance
$S' = (U,\cal S')$ on the same elements as follows. We introduce a set
$s \in \cal S'$ as $s = \cup_{i=1}^p s_i$ for each subcollection
$\{s_1,s_2,\ldots,s_p\} \subseteq \cal S$ of size $p$ where $1 \leq p
\leq \lfloor m/\log m\rfloor$.

\begin{claim}
Number of sets in the new instance $S'=(U,\cal S')$ is $2^{o(m)}$. The
new instance can be constructed in time $2^{o(m)}$.
\end{claim}
\begin{proof}
Recall that the number of sets in the original instance is $\sigma =
m\cdot 2^{\log^{\alpha} m}\cdot \poly\log(m)$ because of the choice
of $\epsilon$. Thus since $p\leq \sigma/2$, the number of sets in
the new instance is
$$\sum_{p=1}^{\lfloor m/\log m\rfloor} {\sigma \choose p} \leq \lfloor
m/\log m\rfloor \cdot {m\cdot 2^{\log^{\alpha} m}\cdot
\poly\log(m)\choose \lfloor m/\log m\rfloor}.$$ We use the inequality
${n\choose p}\leq (ne/p)^p$ to upper-bound this by
$$\lfloor m/\log m\rfloor \cdot {\left(e\cdot 2^{\log^{\alpha} m}
\cdot \poly\log(m) \cdot 2\log m\right)}^{m/\log m}
=2^{O(m/\log^{1-\alpha} m)}=2^{o(m)},$$ as claimed, where we have
$2\log m$ in the first expression (instead of just $\log m$) because
of the floor function and the last equality holds since $0< \alpha<
1$. It is easy to see that the new instance can be created in
$2^{o(m)}$ time.
\end{proof}

\begin{claim}
{\bf Proof of Theorem~\ref{1}}.
The problem is $(r,t)$-$\fpt$-hard for
$r(k) = (\log k)^\gamma$ and $t(k) =
\exp(\exp((\log k)^\gamma))$ for any $1 < \gamma < 1/\alpha$.
\end{claim}
\begin{proof}
Clearly, any optimum will use
as few sets of size (roughly) $m/\log m$
and so the gap between a ``Yes'' instance and a ``No''
hardly changed.
Namely,
$\opt_1$ and that of the new instance $\opt_2$ are related as
$\opt_1 / \lfloor m/\log m \rfloor \leq \opt_2 \leq \lceil
\opt_1/\lfloor m/\log m \rfloor \rceil$. Therefore the gap between
the new optima of a ``yes'' instance and a ``no'' instance continues
to be $c'\log m$ for some constant $c' > 0$ and the new optimum of
the ``yes'' instance is at most $\kappa_T = \lceil \kappa/\lfloor
m/\log m \rfloor \rceil = O(2^{\log^\alpha m} \cdot \poly\log(m))$.
and $\kappa_N$ is $c'\cdot \log m$ larger than that.

Now define two functions $r(k) = (\log k)^\gamma$ and $t(k) =
\exp(\exp((\log k)^\gamma))$ for any $1 < \gamma < 1/\alpha$, as
given in Theorem~\ref{1}.
Note that  $r(\kappa)
=O((\log^{\alpha} m)^\gamma)=o((\log^{\alpha} m)^{1/\alpha})= o(\log
m)$ and $t(\kappa) = 2^{o(m)}$.
Thus this reduction satisfies all the conditions in
Definition~\ref{def:gapred} for Claim~\ref{trivial1} to hold. Thus
$\setcover$ is $(r,t)$-$\fpt$-hard for these functions, proving
Theorem~\ref{1}.
\end{proof}

\subsection{Proof of Theorem \ref{3}}

For proving this theorem we assume:
\begin{conjecture}
There exists a constant $c>0$ and a $\pcp$ of size $m\cdot
\poly\log(m)\poly(1/\epsilon)$, for any $\epsilon$ so that
$\epsilon\geq 1/m^c$.
\end{conjecture}

{\bf Is the conjecture reliable?}

This result was conjectured to hold by Moshkovitz
in a private communication.
Note that there exists already a
$\pcp$ of size even smaller than the above.
In fact in \cite{irit}
a $\pcp$ is presented whose size is
is $m\cdot \poly\log(m)$.
The down size is that the inapproximability of this $\pcp$ \cite{irit}
is $2$. Improving the inapproximability to polylogarithmic
does not seem far fetched.

We now use the above conjecture and show a much stronger
$\fpt$ inapproximability for $\setcover$.
By Corollary \ref{scc},
and the above conjecture we get the following corollary, using
$\epsilon=c^2/\log^2 m$:
\begin{corollary}
There exists a constants $c,c_1>0$ and
a reduction from $3$-$\sat$ to
$\setcover$ so that:
\begin{enumerate}
\item
The number of sets is
$\sigma = m\cdot \poly\log(m)\cdot (1/\epsilon)^{c_1+c_2}
= m\cdot \poly\log(m)$ for some constants $c_1,c_2$.
\item
The number of elements is
$\poly(m)$.
\item
The value of the optimum in ``yes''
instance is {\em exactly}  $\kappa_Y = m \cdot \poly\log(m)$ and that
in the ``no'' instance is at least $c\cdot \log m\cdot \kappa$ with
$c$ some constant $c>0$.
\end{enumerate}
\end{corollary}

{\bf Proof of Theorem \ref{3}:}
Make every collection of sets of size $m/(d\cdot \log\log m)$
one big `collection set', with $d$ a large enough constant.
Here we omit the floor and the ceiling as, in the previous proof, we saw that
they hardly make a difference, and the correction needed is minimal.

The number of sets in the instance is:
$$m\cdot \poly\log(m) \choose \frac{m}{d\cdot \log\log m}$$
and is $2^{o(m)}$ if $d$ is large enough. This is implied by the inequality
${n\choose k}\leq (ne/k)^k$.
The reason for the major improvement is that
the term $2^{\log^{\alpha} m}$
is gone.

After this change, the size of the optimum for a ``yes'' instance
becomes $\poly\log(m)$. Recall that the gap is $c\log m$. Therefore,
the gap can be stated as $\opt^{d'}$ for some $d'<1$.

Let $d''$ be any constant $d''<d'$.
As for the running time, we use
$\opt^{d'}=c\log m$, we get
$2^{\opt^{d''}}=o(m)$ and
$\exp({2^{(\opt^{d''})}})=2^{o(m)}$.
This ends the proof of Theorem \ref{3}.

\subsection{An inapproximability under the Exponential Time Hypothesis  only}
For the (maybe unlikely) case that $\pgc$ will be proved wrong, we
now prove a somewhat weaker inapproximability for $\setcover$
assuming $\eth$ only. This result will remain valid even if  $\pgc$
is disproved.

The following is proved in \cite{dana}.

\begin{theorem}
There exists a constant $c$ and a $\pcp$ of size $m\cdot
2^{\log^{\alpha} m}\cdot \poly\log(m)\poly(1/\epsilon)$, such that
the size of the alphabet is at most $\exp(1/\epsilon)$ and the gap
that can be chosen to be $1/\epsilon$ for any $\epsilon>1/m^c$.
\end{theorem}

The difficulty now is that choosing too large
$\epsilon$ increases the number of sets
very much. Indeed, the number of sets equals
the number of vertices in $\minrep$ and this number is now
$$m\cdot 2^{\log^{\alpha} m}\cdot \poly\log(m)\poly(1/\epsilon)\exp(1/\epsilon).$$

We choose
$\epsilon=\ln 2\cdot \log^{\alpha} m$. Then using a reduction from
$3$-$\sat$ to $\setcover$ described in Corollary \ref{scc} we get:

\begin{corollary}
There exists a constant $d>0$, and a constant
$0<\alpha<1$ and a reduction from $3$-$\sat$ to
$\setcover$ so that:
\begin{enumerate}
\item The number of sets is
$m\cdot 2^{2\log^{\alpha} m}\cdot \poly\log(m).$
\item The number of elements is $\poly(m)$.
\item The gap is $d\cdot \sqrt{\log^{\alpha} m}$.
\item The optimum of a ``yes'' instance does not change, namely,
is
$\opt=2^{\log^{\alpha} m}\cdot \poly\log(m).$
\end{enumerate}
\end{corollary}

The proofs here are simple computations using the new value of
$\epsilon$ plugged in Corollary \ref{scc}. The optimum $\opt$ does
not change because it does not depend on the alphabet. The reason
is, that any optimal solution still takes one vertex from any
supervertex hence the optimum for a ``yes'' instance is still the
number of super vertices.

\noindent{\bf The inapproximability in terms of $\opt$:} The gap is
$d\sqrt{\log^\alpha m}$ for some constant $d$.
$\opt=2^{\log^{\alpha} m} \poly\log(m).$ Thus for some constant $c$,
the problem is $c\cdot \sqrt{\log \opt}$-hard.

\noindent{\bf The time in terms of $\opt$:}
Since $\opt$ did not change we derive exactly the same time as in
Theorem \ref{1}, namely,
$\exp\left (\opt^{(\log^{f} \opt)}\right)$ for the same constant $f>0$,
that appears in Theorem \ref{1}.

This proves Theorem \ref{alone}.

\subsection{Discussion of this specific technique}
This technique alone, combined with the type of $\pcp$ used cannot be
used to prove Fellows conjecture for $\setcover$ because of the
limitation of the $\pcp$. Only a linear reduction from 3-$\sat$ to
$\setcover$ can be used to prove the conjecture. However it is known
(folklore) that such reduction can not exist as it contradicts
$\eth$.

There is a very strong evidence that
a linear $\pcp$ can not exist.
The ultimate $\pcp$ we may expect
(albeit this is not known even for constant $\epsilon$)
is a $\pcp$ of size $m\cdot \poly(1/\epsilon)$ with gap
$1/\epsilon$.
For the choice of $1/\epsilon=\poly\log(m)$
the inapproximability is almost the same as in
Theorem \ref{3}. The running time too
is doubly exponential.
This shows that Theorem \ref{3} got almost the best result possible if we only use
almost linear $\pcp$ and our technique.
It may be that the current state of $\pcp$ theory does not allow the
proof of Fellows conjecture for $\clique$ or $\setcover$. A new type
of $\pcp$, namely, a parameterized version of the $\pcp$ theorem may
be required.

\section{A constant lower bound for Clique in arbitrarily large time in $\opt$}

We use the basic $\pcp$ theorem:
\begin{theorem}[The standard $\pcp$ theorem]
There exists a reduction from any NP-complete language $L$ to
3-$\sat$ so that a ``yes'' instance is mapped to a 3-$\sat$ instance
such that all clauses can be simultaneously satisfied, while a
``no''-instance is mapped to an instance such that at most
$1-\epsilon$ fraction of the clauses can be simultaneously
satisfied. Here $\epsilon>0$ is a constant.
\end{theorem}

We now describe (for the sake of completeness) a totally standard
reduction from 3-$\sat$ to $\clique$.  In the reduction, for each
clause $C$ in the 3-$\sat$ instance, we add a set of seven new
vertices $V_C$ -- one vertex for each of the seven satisfying
assignments to the three variables in the clause. Thus each vertex
corresponds to a partial assignment, i.e., truth-assignment to three
variables. We add an edge between two vertices if the corresponding
partial assignments are `compatible', i.e., they have a common
extension as a full assignment.  Note that if two clauses $C_1$ and
$C_2$ do not share any variables, there is a complete bipartite
graph between the corresponding sets of seven vertices. Note also
that $V_C$ forms an independent set for any clause $C$,
because they, by definition, must disagree on
the value of at least one variable in the clause.

We can combine the $\pcp$ theorem with the reduction described above
and get the following claim using the following standard proof.
\begin{claim}
If a 3-$\sat$ instance, with $n$ variables and $m$ clauses, is a
``yes'' instance, the corresponding $\clique$ instance has a clique
of size $m$. If it is a ``no'' instance, the maximum clique in the
corresponding $\clique$ instance has size at most $(1-\epsilon)m$.
\end{claim}
\begin{proof}
A ``yes'' instance has a satisfying assignment $\pi$. For each
clause $C$, we include the unique vertex in $V_C$ corresponding to
the restriction of $\pi$ onto the variables in $C$, in set $S$. It
is easy to see that $S$ forms a clique of size $m$.

For a ``no'' instance, suppose there is a clique $S$, in the
corresponding $\clique$ instance, of size $\kappa$. Let $\pi$ be any
assignment which is an extension of the partial assignments
corresponding to the vertices in $S$. Now note that $|S \cap V_C|
\leq 1$ for any clause $C$, since $V_C$ is an independent set. Thus
there are $\kappa$ clauses $C_1,\ldots,C_\kappa$ such that $|S \cap
V_{C_i}| = 1$ for all $1\leq i\leq \kappa$. From the definition of
the reduction, it is easy to see that $\pi$ satisfies all these
$\kappa$ clauses. Thus we have $\kappa \leq (1-\epsilon)m$, as
desired.
\end{proof}

We do the following transformation that is a modification of what we
did for for $\setcover$.  The number of vertices in the $\clique$
instance is $7m$.  Let $f(m)$ be {\em any} slowly non-decreasing function of
$m$ such that $f(m) = \omega(1)$.
First, note that we may assume that $m$ is divisible
by $f(m)$ without loss of generality.
Indeed, we need to add fake clauses to the $3$-$\sat$ instance
of the type $(x\vee \bar x\vee z_1),
(x\vee \bar x\vee z_2),\ldots$ so that the number of clauses
added is at most $f(m)$ and we make $m$ divisible by $f(m)$.
Since $f(m)$ is very small compared to $m$, this makes no difference.
We create a new
$\clique$ instance by introducing a vertex for each subset of size
$m/f(m)$ vertices in the old $\clique$ instance.  Such a vertex
is called a `supervertex'.  Two supervertices $A,B$, are
connected by an edge, if $A\cup B$ is a clique, {\em and}
$A\cap B=\emptyset$. The last condition, namely, the fact
that two sets that are connected must be disjoint is not needed in the
$\setcover$ reduction, but it is crucial here.

\begin{claim}
The new instance of the $\clique$ problem has size $2^{o(m)}$.
\end{claim}
\begin{proof}
Using ${n \choose k}\leq (ne/k)^k$, we get that the number of
supervertices is at most $(7e\cdot f(m))^{7m/f(m)} =
\exp(\log(7e\cdot f(m))\cdot 7m/f(m)) = 2^{o(m)}$, since $f(m) =
\omega(1)$. The number of edges in the new $\clique$ instance, being
at most the square of the number of vertices, is also $2^{o(m)}$.
\end{proof}

\begin{claim}
The maximum clique size in any new instance is exactly
$f(m)$. The gap between the clique sizes of the new ``yes''
and ``no'' instances is
$1/(1-\epsilon)$, which implies
$1/(1-\epsilon)$-hardness.
\end{claim}
\begin{proof}
Since the maximum clique size in the old instance is $m$, we get
that the maximum clique size in the new instance
is $f(m)$. Indeed, we can take the optimum clique and divide it into
$m/f(m)$ disjoint sets. By the chosen size these sets are supervertices and
their union is the old optimum clique.
This shows that the new size of the clique is at least $f(m)$.
Since two distinct collection vertices $A$ and
$B$ are adjacent in the new instance, only if $A\cup B$ is a clique,
{\em and} $A,B$ are disjoint, it follows that the largest clique
size of the new ``yes'' instance is exactly $f(m)$
because taking more than $f(m)$ disjoint sets gives a clique of size
larger than $m$, contradicting the fact that $m$ is the maximum
size of the clique. Thus for a yes instance $f(m)$ is the new size of the maximum clique.

The maximum clique in the new ``no'' instance, on the other hand, is
at most $(1-\epsilon)m/(m/f(m))=
f(m)(1- \epsilon )$, otherwise
there would exist a clique in the old instance of size larger than
$(1-\epsilon)m$. The proof is thus complete.
\end{proof}

\begin{claim}
\label{anyt}
The time can be set to be $t(\opt)\cdot n^{O(1)}$ for any
non decreasing function $t$
\end{claim}
\begin{proof}
Since $f(m)$ can be as small as we wish,
we can make the time $t(f(m))$
as small as we want. Let $h(\opt)=2^{o(m)}.$
Selecting
$f(m)=t^{-1}(h(m))$ gives $h(m)=2^{o(m)}$ time.
Since $m,n$ are linearly related here the time can be set to $2^{o(n)}$
for any $t$.
\end{proof}

\subsection{A non constant inapproximability}

Let the graph that we built in previous subsection (whose optimum
for a ``yes'' instance  was $f(m)$) be denoted $H(V,E)$. Recall that
its size is:
$$2^{2\cdot \log(7\cdot e\cdot f(m))\cdot 7m/f(m)}.$$

We now recall the power of a graph $H(V,E)$. We assume the graph is
simple, namely has no loops or parallel edges.

\begin{definition}
The graph $H^k$ has all the
tuples $(v_1,v_2,\ldots,v_k)$
so that any $v_i$ is a vertex of $V$.
The edges are defined as follows.
A tuple $(u_1,u_2,\ldots,u_k)$ is joined to
$(v_1,v_2,\ldots,v_k)$ if and only iff
for $i=1$ to $k$, either $(u_i,v_i)\in E$
or $u_i=v_i$.
\end{definition}

Note that two different vertices in $H^k$ have to differ in at least one tuple
value.

The following theorem is folklore. Let $\omega(H)$ be the size of
the clique in $G$.

\begin{theorem}
$\omega(H^k)=\omega(H)^k$.
\end{theorem}

To get a super constant gap we take the graph $H(V,E)$ of previous
section and raise it to the power $\sqrt{f(m)}$. The choice of
$\sqrt{f(m)}$ is rather arbitrary. Recall that for a ``yes''
instance $\omega(G)=m$, with $m$ the number of clauses in the
$3$-$\sat$ instance and for a ``no'' instance $\omega(G)\leq
(1-\epsilon)m$. Hence $m=\opt$ for a ``yes'' instance. Taking this
graph to the $\sqrt{f(opt)}$ value we get that:

\begin{corollary}
For $H(V,E)^{\sqrt{f(\opt)}}$, the value of the clique for a ``yes''
instance is $f(\opt)^{\sqrt{f(\opt)}}$ and for a ``no'' instance at
most $(1-\epsilon)^{\sqrt{f(\opt)}}\cdot \opt^{\sqrt{f(\opt)}}$.
\end{corollary}

Note that the new size of the graph is:
$$2^{2\cdot \sqrt{f(m)}\log(7\cdot e\cdot f(m))\cdot 7m/f(m)}=2^{o(m)}.$$

In addition, the gap is now
$r(\opt)=(1/(1-\epsilon))^{\sqrt{f(m)}}$. We now describe the gap as
a function of the new optimum. The optimum for a ``yes'' instance is
$\opt'=f(m)^{\sqrt{f(m)}}$. Thus $(\log \opt')^{1/3}=\sqrt{f(m)}$.
Thus the gap in terms of $\opt'$ is:
$r(m)=(1/(1-\epsilon))^{\log^{1/3} \opt'}$.

\begin{claim}
Let $t$ be any non decreasing function and
$r(m)=(1/(1-\epsilon))^{\log^{1/3} \opt'}$.
Then, Clique is $(r,t)$-$\fpt$-hard.
\end{claim}
\begin{proof}
The arguments for $t(\opt)=2^{o(m)}$ follows exactly as
in
Claim \ref{anyt}.
Because the new optimum for a ``yes'' instance $f(m)^{\sqrt{f(m)}}$,
can be made arbitrarily small as well.
\end{proof}

Also, as the new $n$ is $n'=n^{\sqrt{f(m)}}$
we get $n=n'^{1/\sqrt{f(m)}}$. As $f(m)$ can be chosen
arbitrarily small, and $n=7m$, the time as a function of $n$
is $n'^{1/Q(n)}$
for any slowly increasing function $Q(n)$.

\section{FPT Hardness for Minimum Maximal Independent Set}

In this section we obtain hardness for $\mmis$ and thus prove
Theorem~\ref{4}.

We start with 3-$\sat$ instance $I$ with $m$ clauses and $q$ variables.
We assume that
a ``yes'' instance admits a satisfying assignment
and in the case of a ``no'' instance,
any assignment will leave at least one clause
unsatisfied.
We now describe how to build the new graph
$G(I)=(V(I),E(I))$.

\noindent {\bf The building blocks:}
\begin{enumerate}
\item
For every variable $x$ in $C$, we define two vertices $u_x$ and ${\bar
u}_x$. The choice of a vertex $u_x$ represents an assignment True to
$x$ and the choice of ${\bar u}_x$ represents a False assignment to
$x$.

\item
For every clause we add a set $W(C)$ of $q$ copies of the clause.
Namely, $W(C)=\{w^C_1,\ldots,w^C_q\}$.
\end{enumerate}
Intuitively, we want to create a $\setcover$-like instance
in which variables are sets and clauses are elements
and a variable $u_x$ covers $C$ if $x\in C$
and ${\bar u}_x$ covers $C$ if ${\bar x}\in C$.

\noindent{\bf Supervertices:}
Similar to our construction for $\setcover$ and $\clique$,
we define a new graph $H(I)$ with supervertices that are
collections of vertices of the type
$u_x,{\bar u}_x$.

Let $f(q)$ be {\em any} slowly increasing function of $q$ such that
$f(q)=\omega(1)$ and assume, by adding dummy clauses if needed, that
$f(q)$ divides $q$.  The supervertices of $V(I)$, denoted by $v_S$,
correspond to subsets $S\subseteq \{u_x\mid x\in C\}\cup \{{\bar
u}_x\mid x\in C\}$ satisfying the following two conditions:
\begin{enumerate}
\item $|S|=q/f(q)$,
\item $S$ does not contain both $u_z,{\bar u}_z$ for any variable $z$
(i.e., a set $S$ does not contain a ``contradiction'' in the truth
value assignment).
\end{enumerate}

\noindent{\bf Edges between two supervertices:} Introduce an edge
between $v_{S_1}$ and $v_{S_2}$ if and only if there exists some
variable $x$ so that either $u_x$ or ${\bar u}_x$ belongs to $S_1$
{\bf and} either $u_x$ or ${\bar u}_x$ belongs to $S_2$ Note that the
above gives four cases in which $v_{S_1},v_{S_2}$ are connected.

\noindent{\bf Edges between supervertices and $W(C)$ vertices:}
Introduce edges as follows:
\begin{enumerate}
\item  If a variable $x\in C$, any supervertex that contains
the vertex $u_x$ is connected to all vertices of $W(C)$.
\item If a variable $\bar x\in C$, any supervertex that contains
${\bar u}_x$ is connected to all vertices of $W(C)$.
\end{enumerate}

\noindent{\bf Example:}
Say for example
$C=(x\vee \bar z\vee w)$.
Then any supervertex that contains $u_x$
is connected to all the copies of $W(C)$. Also, every supervertex
that contains ${\bar u}_z$ or
$u_w$ is connected to all the copies
of $W(C)$.


\begin{claim}
Total number of vertices in $H(I)$ is $2^{o(q)}+qm$. The instance
$G(I)$ can be constructed in time $2^{o(q)}$.
\end{claim}
\begin{proof}
The total number of vertices in $H(I)$ of type $v_S$ for $S\subset
A$ is at most ${q \choose  q/f(q)}<
(qe/(q/f(q)))^{q/f(q)} < 2^{o(q)}$.  Here we again use the
inequality ${q \choose k} \leq (qe/k)^k$. The number of vertices of
type $W(C)$ for a clause $C$ is $qm$.
\end{proof}

{\bf Building an} $\mmis$ {\bf of size $f(q)$ for a ``yes''
instance:}

\begin{enumerate}
\item
Start with the set $X=\{u_x\mid x \mbox{ is a literal}\}$.
This set contains for every variable
its vertex copy that corresponds
to a True assignment.
\item
Decompose $X$ to $f(q)$ {\em pairwise disjoint} sets each
containing $q/f(q)$ vertices. Let these sets
be $S_1,S_2,\ldots,S_{q/f(q)}$. We want to derive sets so that
$v_{S_i}$ is a feasible $\mmis$, which is of course
not the case so far (because not all $W(C)$ are covered).
\item
We now modify sets $S_i$ to obtain sets $T_i$ as follows.  Fix a
satisfying assignment $\tau$ to the variables. We start by setting
$T_i=S_i$ for all $i$. If $\tau(x)$ is False, then for the unique $i$
so that $u_x\in T_i$, remove $u_x$ from $T_i$ and add ${\bar u}_x$ to
$T_i$.  This is done for all variables.  The final $T_i$ sets are
called the {\em assignment sets}.  Our solution will be ${\cal
I}=\{v_{T_i}\mid T_i \mbox{ is an assignment set}\}$.
\end{enumerate}

\begin{claim}
\label{noedges}
The set $\{v_{T_i}\}$ is independent in $H(I)$.
\end{claim}
\begin{proof}
For the vertices $v_{T_i}, v_{T_j}$ with $i\neq j$ to be connected it
must be that some $x$ so that either $u_{x}$ or ${\bar u}_x$ belongs
to $T_i$ and either $u_{x}$ or ${\bar u}_x$ belongs to $T_j$.
Clearly, this implies that $u_x\in S_i\cap S_j$. This is a
contradiction to the fact that the sets $\{S_p\}$ are pairwise
disjoint.
\end{proof}

\begin{claim}
\label{yesm}
The $f(q)$ vertices $\{v_{T_i}\}$ defined above form a dominating set
in $H(I)$.
\end{claim}
\begin{proof}
We first show each vertex in $W(C)$ is adjacent to some vertex
$v_{T_i}$.  Note that $\tau$ satisfies all clauses $C$.  One
possibility is that $\tau(x)$ is True and $x\in C$.  Thus the {\em
unique} assignment set $T_i$ that contains $u_x$ is connected to all
the copies $W(C)$ of $C$.  Alternatively, if $\tau(x)$ is False and
${\bar x}\in C$, the unique $T_i$ that contains ${\bar u}_x$ is
connected to all copies of $W(C)$.

We now show that ${\cal I}$ dominates every supervertex not in ${\cal
I}$. Let $v_{S}$ be a vertex of $H(I)$ that does not belong to ${\cal
I}$. Pick an arbitrary variable $x$ so that either $u_x\in S$, or
${\bar u}_x\in S$. By construction there is some assignment set
$T_i\in {\cal I}$ that contains $u_x$ or ${\bar u}_x$.  In all the
fours cases above, by definition, there is an edge between $v_{T_i}$ and $v_S$.
\end{proof}

Thus we just proved the following corollary.
\begin{corollary}
The ``yes'' instance admits a solution of size $f(q)$.
\end{corollary}

\begin{claim}
\label{nom}
For a no instance the minimum $\mmis$ is of size larger than $q$.
\end{claim}
\begin{proof}
Let ${\cal S}$ be the optimum $\mmis$ of the ``no'' instance.  Note
that all super vertices chosen by the optimum have to be
consistent. Namely, we can not have $u_x$ belonging to one set $T_i$
in ${\cal S}$ and ${\bar u}_x$ to some $T_j\in {\cal S}$ because this
will imply an edge between $v_{T_i}$ and $v_{T_j}$ and a
contradiction.  In particular, this implies that vertices
$\{v_{T_i}\}$ represent a (maybe partial) truth assignment to the
variables.

Since we are dealing with a no instance, there must be a clause $C$
that is not satisfied by this partial assignment. This means that none
of the vertices that correspond to literals that satisfy $C$ are in
any set of ${\cal S}$.  For example if $C=(x\vee \bar z\vee w)$ then
there may be one set related to $x$ but it contains ${\bar u}_x$,
because the assignment does not satisfy $C$.  There may be one set
related to $z$, but it contains $u_z$, and there may be a set for $w$,
but it contains ${\bar u}_w$. This means that the $q$ copies $W(C)$
must be present in ${\cal S}$, since it is a {\em maximal} independent
set. Thus the size of ${\cal S}$ is at least $q$.
\end{proof}

\begin{theorem}
Assuming the $\eth$,  $\mmis$ problem is
$(r,t)$-hardness, for any $r,t$.
\end{theorem}
\begin{proof}
Since the new optimum for a yes instance is $f(q)$ where $f$ is an
arbitrarily slow growing function.  For any given functions $r$ and
$t$, we can make sure that $r(f(q))<q/f(q)$ and $t(f(q))=2^{o(q)}$.

Note that by Claims \ref{yesm} and \ref{nom}, the gap between ``yes''
and ``no'' instances is larger than $q/f(q)$. If there existed an
$(r,t)$-$\fpt$-approximation for $\mmis$, we could distinguish
between a ``yes'' and a ``no'' instance of $3$-$\sat$ in time $2^{o(q)}$,
contradicting the ETH.
\end{proof}

\section{Open problems}

We can look for more problems that are $(r,t)$-$\fpt$-hardness.
Given a problem we first should check if there is no easy
$g(\opt)$ approximation. many such problems are presented
in \cite{CHK}. All the problems in question were
highly inapproximable and W[1]-hard
Since we believe that such an approximation is not possible
for  $\clique$
we ask:

\noindent{\bf Open Problem 1:}
Which $W[1]$ optimization problems are $(r,t)$-$\fpt$-hard for various functions $r$ and $t$?

One interesting thing is that all problems presented in \cite{CHK}
were only $W[1]$-hard.

\noindent
{\bf Open Problem 2:}
Is there an $f(\opt)$ approximation for any $W[2]$-hard
problem, that runs in polynomial or $\fpt$ time?

It may be hard to prove that such a thing does not exist as it will
imply  $W[1] \neq W[2]$. Still, this suggests an {\em indirect} way to
try and prove this important and widely believed conjecture.

Another way to make the optimum smaller is by taking a random sample.
This technique fails miserably for $\setcover$, but also
for $\clique$ for which we had hopes that this will succeed.

\noindent
{\bf Open problem 3:} Is there a $W[1]$ or $W[2]$ hard
problem, so that we can decrease $\opt$ by random sampling,
keep a large enough gap,
and imply a super-exponential hardness
in $\opt$ due to the optimum becoming smaller?

\bibliographystyle{abbrv}

\end{document}